\documentclass{llncs}

\usepackage{times}
\usepackage{graphicx}
\usepackage{amssymb}
\usepackage{amsmath}

\usepackage{url}
\usepackage{xspace}
\newcommand{\numSAT}{\textsc{\#Sat}\xspace}
\newcommand{\ASP}{\textsc{Asp}\xspace}
\newcommand{\SAT}{\textsc{Sat}\xspace}

\newcommand{\QBF}{\textsc{Qbf}\xspace}
\newcommand{\pname}[1]{\textsc{#1}\xspace}

\usepackage{xcolor}

\usepackage{hyperref}
\usepackage{multirow}

\usepackage{etoolbox}
\patchcmd{\thebibliography}{\chapter*}{\section*}{}{}

\usepackage{graphicx}
\graphicspath{{fig/}}
\DeclareGraphicsExtensions{.pdf,.png,.jpg}
\usepackage{subfig}
\usepackage{booktabs}

\newenvironment{changemargin}[2]{%
\list{}{\rightmargin#2\leftmargin#1
\parsep=0pt\topsep=0pt\partopsep=0pt}
\item[]}
{\endlist}

\newenvironment{indented}{\begin{changemargin}{1cm}{0cm}}{\end{changemargin}}

\let\phi\varphi
\let\epsilon\varepsilon

\renewcommand{\models}{\vDash}

\newcommand{\calA}{\mathcal{A}}

\newcommand{\calC}{\mathcal{C}}

\newcommand{\calR}{\mathcal{R}}

\newcommand{\calT}{\mathcal{T}}

\newcommand{\card}[1]{\left|#1\right|}
\newcommand{\CCard}[1]{\|#1\|}

\newcommand{\Nat}{\mathbb{N}} 

\newcommand{\algo}[1]{\ensuremath{\mathsf{#1}}}

\newcommand{\NP}{\ensuremath{\textsc{NP}}}
\newcommand{\co}{\ensuremath{\textsc{co}}}

\newcommand{\bigO}[1]{\ensuremath{{\mathcal O}(#1)}}

\newcommand{\cnt}{\#}

\newcommand{\tw}[1]{\mathit{tw}(#1)}

\newcommand{\mods}[1]{\mathit{models}(#1)}

\newcommand{\nbody}[1]{{\mathit{B}^-(#1)}}
\newcommand{\pbody}[1]{{\mathit{B}^+(#1)}}
\newcommand{\head}[1]{{\mathit{H}(#1)}}

\newcounter{cefalo}

\newcounter{cefalocont}

\def\qed{\hfill{\qedboxempty}      
  \ifdim\lastskip<\medskipamount \removelastskip\penalty55\medskip\fi}

\def\qedboxempty{\vbox{\hrule\hbox{\vrule\kern3pt
                 \vbox{\kern3pt\kern3pt}\kern3pt\vrule}\hrule}}

\def\qedfull{\hfill{\qedboxfull}   
  \ifdim\lastskip<\medskipamount \removelastskip\penalty55\medskip\fi}

\def\qedboxfull{\vrule height 4pt width 4pt depth 0pt}

\hypersetup{pdfinfo={
  Title={Counting Answer Sets via Dynamic Programming}
  Author={Johannes Fichte, Markus Hecher, Michael Morak, Stefan Woltran}
}}

\title{Counting Answer Sets via Dynamic Programming}

\author{Johannes Fichte \and Markus Hecher \and Michael Morak \and Stefan
Woltran}

\institute{TU Wien, Vienna, Austria\\
  lastname@dbai.tuwien.ac.at}

\begin{document}

\maketitle

\begin{abstract}
  While the solution counting problem for propositional satisfiability (\#SAT)
  has received renewed attention in recent years, this research trend has not
  affected other AI solving paradigms like answer set programming (ASP).
  Although ASP solvers are designed to enumerate all solutions, and counting
  can therefore be easily done, the involved materialization of all solutions
  is a clear bottleneck for the counting problem of ASP (\#ASP).
  In this paper we propose dynamic programming-based \#ASP algorithms that
  exploit the structure of the underlying (ground) ASP program. Experimental
  results for a prototype implementation show promise when compared to existing
  solvers.
\end{abstract}

\section{Introduction}\label{sec:introduction}
With the rise of efficient solvers, model counting for the
propositional satisfiability problem (\#SAT)
\cite{siamcomp:Valiant79,aaai:BayardoS97} has received renewed
attention in recent years (cf.\ e.g.,
\cite{ijcai:ChakrabortyFMV15,sat:AzizCMS15}). Knowing the number of
models of a
propositional formula is a useful measurement, and can be used, inter
alia, in the areas of machine learning, probabilistic reasoning, statistics, and
combinatorics \cite{ai:Roth96,aaai:SangBK05,jair:DomshlakH07,aaai:XueCD12}.
Various systems 
have been implemented that
solve the \#SAT problem; see~e.g.,~\cite{sat:Thurley06,sat:SangBBKP04}.

Similar strides in efficiency have also been made in answer set
programming (ASP) regarding the model existence
problem~\cite{iclp:GelfondL88,cacm:BrewkaET11}, where efficient
solvers are now readily
available,~e.g.,~\cite{GebserKaufmannSchaub12a,lpnmr:AlvianoDFLR13}. ASP
is a rule-based language that has found great success as it allows
users to specify intuitive, fully-declarative problem descriptions,
and is used in both industry and research. When using SAT solvers to
evaluate a problem, the problem usually has to be rewritten into a SAT
formula. While such SAT rewritings require a specialized algorithm for
each particular problem, ASP as a rule-based formalism allows for the
declarative specification of problem statements. The actual problem
instance can then simply be given as an input database of ground
facts.

For instance, a simple graph 2-colorability problem can be stated,
using two rules, as follows:
  \begin{eqnarray*}
	\mathit{color}(V, \mathit{red}) \vee \mathit{color}(V, \mathit{blue})
	\leftarrow \mathit{vertex}(V)\\
	\bot \leftarrow \mathit{edge}(V1, V2), \mathit{color}(V1, C),
	\mathit{color}(V2, C)
  \end{eqnarray*}
  Together with a graph, given as a set of facts of the form $\mathit{edge}(x,
  y)$ and $\mathit{vertex}(v)$, each answer set represents exactly one valid
  coloring of the graph. 
%
Evaluating ASP programs like the one 
above is usually a two-step
process. First, a \emph{grounder} instantiates the program, replacing all
variables by domain constants, and then a \emph{solver} evaluates the ground 
program and computes the answer sets.  While for SAT the model existence problem
is \NP-complete, the problem of evaluating ground (disjunctive) ASP programs is
located on the second level of the polynomial hierarchy. Thus, ASP allows for
efficient encodings of problems of higher complexity that typically arise in AI,
like circumscription or diagnosis. 
Opposite to 
standard SAT solvers which simply decide the problem
or deliver one (counter-)model,
ASP systems are tailored to
enumerate all answer sets.
Due to this fact, the answer set counting
problem (\#ASP) has received far less attention than the \#SAT problem. However,
materializing all answer sets can be expensive and is not necessarily required
for counting.

It is the aim of this paper to propose and evaluate a dynamic
programming-based answer set counting algorithm that exploits the
structure of the given (ground) input ASP program and avoids the
costly materialisation of all answer sets.  The importance of
evaluating such parameterized algorithms in practice has recently been
stressed~\cite{paracompnews:Gutin15}.  Several works have shown that
such an approach works well for \#SAT, guaranteeing both favorable
theoretical runtime bounds \cite{dam:FischerMR08,jda:SamerS10}, as
well as good practical performance in prototype implementations
\cite{jair:LiPB11,jair:SaetherTV15}.
Jakl et al.~\cite{ijcai:JaklPW09} and Morak et
al.~\cite{jelia:MorakPRW10} have introduced dynamic programming
algorithms for deciding the answer set existence problem in linear
time in the size of the input, where the enumeration of answer sets
can be done with linear delay. This is accomplished by exploiting the
structure of the given program via tree decompositions of its
incidence graph. A tree decomposition of a graph (roughly) tries to
arrange the graph into a tree by combining cyclic parts of the graph
into single tree nodes. If the size of these tree nodes can be bounded
by a (small) constant, then the problem can be evaluated efficiently
by traversing the tree decomposition in a bottom-up manner, evaluating
the answer set existence problem only locally for each node and
combining the partial solutions.
In this paper, we investigate how this idea can be used for counting without
answer set materialization, and propose additional variations of the algorithm
based on different graph representations of the program.

\paragraph{Contributions.} The main contributions of this paper are:
\begin{enumerate}
  \item We use dynamic programming 
	on tree decompositions
	to solve the answer set counting problem for ASP. Three
	versions of the algorithms are proposed, based in part on different graph
	representation of the input.

  \item We show that the algorithm exhibits favorable theoretical runtime
	bounds.

  \item We provide a prototype implementation of our proposed algorithms, and
	give an experimental performance analysis and evaluation, comparing the
	solver to several existing solutions.
\end{enumerate}

The remainder of the paper is structured as follows: In
Section~\ref{sec:preliminaries}, we give some preliminaries on ASP and tree
decompositions. Section~\ref{sec:results} gives an overview of the general
principles of dynamic programming algorithms on tree decompositions and then
proceeds to give the proposed answer set counting algorithms. 
Finally, a prototype implementation is evaluated via
experiments in Section~\ref{sec:evaluation}. We close with some concluding
remarks in Section~\ref{sec:conclusions}.


\section{Preliminaries}\label{sec:preliminaries}

\paragraph{Answer Set Programs.} A \emph{ground answer set program} (or
program, for short) is a pair $\Pi = (\calA,\calR)$, where $\calA$ is a set of
propositional atoms and $\calR$ is a set of rules of the form:
\begin{equation}\label{eq:rule}
  a_1\vee \cdots \vee a_l \leftarrow a_{l+1}, \ldots, a_m, \neg a_{m+1}, \ldots,
  \neg a_n,
\end{equation}
where $n \geq m \geq l$ and $a_i \in \calA$ for all $1 \leq i \leq n$. A rule
$r \in \calR$ of form~\eqref{eq:rule} consists of a head $\head{r} = \{
a_1,\ldots,a_l \}$ and a body 
given by
$\pbody{r} = \{a_{l+1},\ldots,a_m \}$ and $\nbody{r} = \{a_{m+1},\ldots,a_n \}$.
A set $M \subseteq \calA$ is a called a model of $r$, if $\pbody{r} \subseteq M$
together with $\nbody{r} \cap M = \emptyset$ implies that $\head{r} \cap M \neq
\emptyset$.  We denote the set of models of $r$ by $\mods{r}$ and the models of
a program $\Pi= (\calA,\calR)$ are given by $\mods{\Pi} = \bigcap_{r \in \calR}
\mods{r}$.

The reduct $\Pi^I$ of a program $\Pi$ with respect to a set of atoms $I
\subseteq \calA$ is the program $\Pi^I = \left(\calA,\left\{r^I \mid r \in
\calR, \nbody{r} \cap I = \emptyset)\right\}\right)$, where the reduct $r^I$ of
a rule $r$ is the same rule without negative body, i.e., $\head{r^I} =
\head{r}$, $\pbody{r^I} = \pbody{r}$, and $\nbody{r^I} = \emptyset$.
Following~\cite{iclp:GelfondL88}, $M \subseteq \calA$ is an \emph{answer set} of a
program $\Pi = (\calA,\calR)$ if $M \in \mods{\Pi}$ and for no $N \subsetneq M$,
we have $N \in \mods{\Pi^M}$. The \emph{consistency problem} of ASP (decide
whether there exists an answer set for a given program $\Pi$) is
$\Sigma^P_2$-complete~\cite{amai:EiterG95}.

\paragraph{Fixed-Parameter Tractability.} We briefly recall the basic
notions of fixed-parameter tractability. For more detailed information we
refer to other sources, e.g.\ \cite{book:DowneyF13}. A \emph{parameterized
problem}~$L$ is a subset of~$\Sigma^* \times \Nat$ for some finite alphabet
$\Sigma$. For an instance~$(I,k) \in \Sigma^* \times \Nat$ we call $I$ the
\emph{main part} and $k$ the \emph{parameter}. A problem $L$ is
\emph{fixed-parameter tractable (FPT)} if there exists a computable function $f$
and a constant $c$, such that there exists an algorithm that decides whether
$(I,k) \in L$ in time~$\bigO{f(k) \cdot \CCard{I}^c}$ where $\CCard{I}$ denotes
the size of~$I$. Such an algorithm is called an \emph{fpt-algorithm}. A well
studied parameter is the so-called treewidth.

\paragraph{Tree Decomposition and Treewidth.} A \emph{tree decomposition} of a
graph $G = (V,E)$ is a pair $\calT = (T, \chi)$, where $T$ is a rooted tree and
$\chi$ is a labelling function with $\chi(t) \subseteq V$---we call $\chi(t)$
the \emph{bag} of $t$---such that the following holds: (i) for each $v \in V$,
there exists a $t \in T$, such that $v \in \chi(t)$; (ii) for each $\{v,w\} \in
E$, there exists a $t \in T$, such that $\{v, w\} \subseteq \chi(t)$; and (iii)
for each $r, s, t\in T$, such that $s$ lies on the path from $r$ to~$t$, we have
$\chi(r) \cap \chi(t) \subseteq \chi(s)$.  The \emph{width} of a tree
decomposition is defined as the cardinality of its largest bag minus one. The
\emph{treewidth} of a graph $G$, denoted by $\tw{G}$, is the minimum width over
all tree decompositions of~$G$. For arbitrary but fixed $w \geq 1$, it is
feasible in linear time to decide if a graph has treewidth $\leq w$ and, if so,
to compute a tree decomposition of width $w$.

A tree decomposition $(T,\chi)$ is called \emph{normalized} (or \emph{nice})
\cite{book:Kloks94}, if (i) each $t \in T$ has at most two children; (ii) for
each $t \in T$ with two children $t'$ and $t''$, it holds that $\chi(t) =
\chi(t') = \chi(t'')$; and (iii) for each $t \in T$ with exactly one child $t'$,
the bags~$\chi(t)$ and $\chi(t')$ differ in exactly one element. 
Every tree decomposition can be normalized
in linear time without increasing its width \cite{book:Kloks94}. We 
assume
w.l.o.g.\ that normalized tree decompositions have root and leaf nodes whose
bags are empty.

\paragraph{Dynamic Programming on Tree Decompositions.} Dynamic programming is
an established technique in the toolkit of parameterized
complexity~\cite{book:Niedermeier06}. Especially, for problems parameterized by
treewidth, dynamic programming algorithms on tree decompositions have been
applied to many graph problems~\cite{mfcs:Bodlaender97}, as well as problems in
Logic and Artificial Intelligence~\cite{lpnmr:PichlerRW09}. Dynamic algorithms
on tree decompositions all share a common structure. The tree decomposition is
traversed from the leaf nodes to the root node. At each node a subproblem is
solved which consists of the part of the problem instance that is induced by the
content of the bag of the current node. This results in a set of partial
solutions (called \emph{tuples}) which is propagated from the child nodes to the
parent node.  From these, the parent node then calculates the partial solutions
its induced subproblem. Finally, at the root node there is a correspondence
between the partial solutions of the root node and the solutions of the whole
problem instance. An appropriate data structure to represent the partial
solutions must be devised: this data structure must contain sufficient
information to compute the representation of the partial solutions at each node
from the corresponding representation at the child node(s). In addition, to
ensure efficiency, the size of the data structure should only depend on the size
of the bag (and not on the size of the entire problem instance).

\paragraph{Tree Decompositions of Logic Programs.} To build tree decompositions
for ground answer set programs $\Pi = (\calA,\calR)$, we use two types of
graph representations:
(a) The \emph{incidence graph} $G_{\mathit{inc}}(\Pi)$ of
$\Pi$ is an undirected, bipartite graph $(\calA \cup \calR, E)$, where $E$
contains an edge $(a, r)$, iff atom $a$ occurs in rule $r$ of $\Pi$; and
(b) the \emph{primal graph} $G_{\mathit{prim}}(\Pi)$ of $\Pi$
is an undirected graph $(\calA, E)$, where $E$ contains an edge $(a_i, a_j)$,
iff there exists an $r \in \calR$, such that both $a_i$ and $a_j$ appear in $r$.
A tree decomposition of such a graph representation of a program $\Pi$ is called
a tree decomposition of $\Pi$.
The treewidth of an incidence graph of a program is at most the
treewidth of its primal graph plus 1.

For normalized tree decompositions of programs, we can distinguish between six
types of nodes: \emph{leaf} (LEAF), \emph{join} (JOIN), \emph{atom introduction}
(AI), \emph{atom removal} (AR), \emph{rule introduction} (RI), and \emph{rule
removal} (RR) node. A node is a leaf node if it has no child node. A node is a
join node if it has two child nodes. A node is an atom introduction node if its
bag contains one additional atom compared to the bag of its child node.
Similarly, a node is an atom removal node, if its bag contains one less atom
compared to the bag of its child node. Rule introduction and rule removal nodes
are defined analogously to AI and AR nodes, where the difference between bags is
an added or removed rule instead of an atom. The last four types, AI, AR, RI,
and RR, will often be augmented with the element $e$ (either an atom or a rule)
which is removed or added compared to the bag of the child node. For example
AI($a$) denotes that atom $a$ is added and RR($r$) denotes that rule $r$ is
removed. Notice that primal graph tree decompositions cannot contain RI and RR
nodes.

When we refer to tree decompositions of programs, we will use the
following notation: Let $\calT = (T, \chi)$ be a tree decomposition of
program $\Pi = (\calA, \calR)$ and let $t \in T$ be a node. We denote
the atoms occurring in~$\chi(t)$ with $A_t$ and the rules occurring
in~$\chi(t)$ with $R_t$. These notions naturally extend from nodes $t$
to subtrees $T_t$ of $T$ (rooted at~$t$). We say that a set of atoms
$M$ entails a rule $r$ w.r.t.\ a tree node~$t$, denoted
$M \models_t r$, iff $M \in \mods{r|_t}$, where $r|_t$ is the rule
obtained from $r$ by removing all literals formed by atoms that are
not contained in $\chi(t)$. For example, a rule where all literals are
removed thus becomes $\bot \leftarrow \top$, that is, a contradiction.

\paragraph{Counting Complexity.}

Investigating the complexity of counting problems was initiated in
\cite{siamcomp:Valiant79}. Formally, a
\emph{counting problem} is presented using a \emph{witness} function which for
every input~$x$ returns a set of \emph{witnesses} for~$x$.  A \emph{witness}
function is a function $w\colon \Sigma^* \rightarrow \Gamma^{*,<\infty}$,
where~$\Sigma$ and~$\Gamma$ are two alphabets; $\Gamma^{*,<\infty}$ are all
finite 
strings
in~$\Gamma^*$. A \emph{counting problem} is then defined as
follows: given $x\in \Sigma^*$, find the cardinality $\card{w(x)}$.
For a standard complexity class $\calC$, define $\cnt\calC$ as the class of all
counting problems whose witness function~$w$ satisfies (a) there is a polynomial
$p(n)$ such that for every $x \in \Sigma^*$ and every $y \in w(x)$ we have
$\card{y} \leq p(\card{x})$; and (b) the problem ``given~$x$ and~$y$, is $y \in
w(x)$?'' is in $\calC$; see also \cite{sigact:HemaspaandraV95}. 


\section{Counting Answer Sets}\label{sec:results}

In this section, we introduce our \#ASP dynamic programming
algorithms. Before doing so, 
we 
give a brief complexity-theoretic discussion of the \#ASP
problem. 

\subsubsection{Complexity of \#ASP}

Before introducing the actual algorithms, we provide the following
straight-forward result that shows that (under standard complexity-theoretic
assumptions) the counting problem we consider here is strictly harder that
\#SAT. In fact the following result is a corollary of existing results that deal
with the complexity of evaluating logic programs and the complexity of counting
subset-minimal models of CNF formulas:

\begin{theorem}
  The counting problem \#ASP is 
	\cnt\co\NP-complete
\end{theorem}

\begin{proof}
  Membership follows from the fact that, given a program $P$ and an
  interpretation~$I$, checking whether $I$ is an answer of $P$ is
  \co\NP-complete, see e.g.\ \cite{KochL99}. Hardness is  a direct consequence of
  \cnt\co\NP-hardness for the problem of counting subset minimal models of a CNF
  formula \cite{DurandHK05}, since answer sets of negation-free programs and
  subset-minimal models of CNF formulas are essentially the same objects.
\end{proof}

We note that the counting complexity of ASP programs including optimization
statements (i.e., where only optimal answer sets are counted w.r.t.\ a cost
function) is slightly higher; exact results can be established 
employing hardness results from~\cite{HermannP09}.

\subsubsection{Counting Algorithms for ASP}

In order to simplify the presentation, we start by giving only the
decision version of the algorithms, and extend them to counting
algorithms later. We assume that, as input, the algorithm is given an
answer set program $\Pi$, and a normalized tree decomposition $\calT$
of $\Pi$. At each node $t$ of the tree decomposition, the algorithm
will compute a set $\tau_t$ of tuples that represent the partial
solutions. Given a node $t$, we denote by $t'$ its first child, and by
$t''$ its second child. Thus, a full specification of the algorithm is
given by describing how the set $\tau_t$ is derived from the sets
$\tau_{t'}$ and $\tau_{t''}$.

We next present three dynamic programming algorithms. The first
algorithm, \textsf{INC}, works on a tree decomposition of the incidence graph
$G_\mathit{inc}(\Pi)$. The other two algorithms, \textsf{PRIM} and
\textsf{INVPRIM}, use the primal graph $G_\mathit{prim}(\Pi)$.

\paragraph{\textsf{INC} Algorithm.} The first algorithm, based on the
incidence graph, is given in Figure~\ref{fig:inc}. It contains a
specification of how, for each node type of the tree decomposition,
the set~$\tau_t$ at tree node $t$ can be derived from the
set~$\tau_{t'}$ and $\tau_{t''}$ of its child nodes. A tuple in such a
set is a triple $\langle M, S, C \rangle$, where $M$ represents a
truth assignment of the atoms in~$\chi(t)$, $S$ is a set of rules in
$\chi(t)$ which are already satisfied, and $C$ is a set of
certificates, that is, pairs $(A, R)$ of sets of atoms and rules,
where $A \subset M$ represents a potential counter-model with respect
to~$\Pi^M$. The existence of such a tuple in~$\tau_t$ witnesses the
existence of a partial answer set for the program induced by the
subtree rooted at $t$. Note that the decision version of the
\textsf{INC} algorithm below restates the algorithm
from~\cite{ijcai:JaklPW09}.

Intuitively, the algorithm works as follows: In an atom introduction node $t$
for atom~$a$, partial solutions of child nodes are extended in two ways: $a$ is
either set to true and added to the set~$M$ in the tuple, or set to false and
not added to~$M$. For each rule in~$\chi(t)$, we verify whether the rule is
satisfied by the choice on $a$. Also, all possible subsets of the new set $M$
are considered as certificates. By connectedness of tree decompositions, once a
rule removal node removes a rule $r$, we can discard all tuples where $r$ is not
yet satisfied, since $r$ will never appear again in an ancestor node. In join
nodes, on agreeing sets $M$, the satisfied rules of the left child tuples and
right child tuples are merged, since all of them are already satisfied. The
certificates are updated in a similar manner. Continuing all the way up to the
(empty) root node, this process guarantees that a surviving tuple witnesses that
there is an assignment $M$ to all atoms of $\Pi$ that satisfies all the rules.
Further, if the set of certificates of such a tuple is empty, then this
witnesses that there does not exist a subset of $M$ that is a model of the
reduct $\Pi^M$. For each node type of the tree decomposition, the construction
of the set of tuples $\tau_t$ of node $t$ from the tuple sets $\tau_{t'}$ and
$\tau_{t''}$ of its child nodes $t'$ and $t''$, which embodies the intuitive
idea above, is given above. Note that we specify the set $\tau_t$ directly.
Checking whether there exists an answer set is equivalent to checking whether,
after a bottom-up traversal of the tree decomposition, the root node contains
the tuple $\langle \emptyset, \emptyset, \emptyset \rangle$.

\begin{figure}
\small
\begin{eqnarray*}
  AI(a): & \{ \langle & M, S \cup \{ r \mid r \in R_t, M \models_t r \},\\
  && \{ (A, R \cup \{ r \mid r \in R_t, A \models_t r^M \}) \mid (A, R) \in
  C \} \rangle\\
  & \mid & \langle M, S, C \rangle \in \tau_{t'} \} \cup\\
  & \{ \langle & M, S \cup \{ r \mid r \in R_t, M \models_t r \},\\
  && \{ (A, R \cup \{ r \mid r \in R_t, A \models_t r^M \}) \mid (A, R) \in
  C \} \cup\\
  && \{ (A, R \cup \{ r \mid r \in R_t, A \models_t r^M \})\\
  && \hspace{2.5em}\mid (A', R) \in C, A = A' \cup \{ a \} \} \cup\\
  && \{ (M', \{ r \mid r \in R_t, M' \models_t r^M \}) \} \rangle\\
  & \mid & \langle M', S, C \rangle \in \tau_{t'}, M = M' \cup \{ a \} \}\\
  \\[-2ex]
  AR(a): & \{ \langle & M \setminus \{ a \}, S,\\
  && \{ (A \setminus \{ a \}, R) \mid (A, R) \in C \} \rangle\\
  & \mid & \langle M, S, C \rangle \in \tau_{t'} \}\\
  \\[-2ex]
  RI(r): & \{ \langle & M, S \cup \{ s \mid s = r, M \models_t s \} \\
  && \{ (A, R \cup \{ s \mid s = r, A \models_t s^M \}) \mid (A, R) \in C
  \} \rangle\\
  & \mid & \langle M, S, C \rangle \in \tau_{t'} \}\\
  \\[-2ex]
  RR(r): & \{ \langle & M, S \setminus \{ r \}\\
  && \{ (A, R \setminus \{ r \}) \mid (A, R) \in C, r \in R \} \rangle\\
  & \mid & \langle M, S, C \rangle \in \tau_{t'}, r \in S \}\\
  \\[-2ex]
  JOIN: & \{ \langle & M, S' \cup S''\\
  && \{ (A, R' \cup R'') \mid (A, R') \in C', (A, R'') \in C''\} \cup\\
  && \{ (A, R \cup S'') \mid (A, R) \in C', A = M \} \cup\\
  && \{ (A, R \cup S') \mid (A, R) \in C'', A = M \} \rangle\\
  & \mid & \langle M, S', C' \rangle \in \tau_{t'}, \langle M,
  S'', C'' \rangle \in \tau_{t''} \}\\
  \\[-2ex]
  LEAF: && \{ \langle \emptyset, \emptyset, \emptyset \rangle \}
\end{eqnarray*}
\normalsize
\vspace{-0.4cm}
\caption{\textsf{INC} Algorithm.}
\label{fig:inc}
\end{figure}

\paragraph{\textsf{PRIM} and \textsf{INVPRIM} Algorithms.} A simplification of
the above \textsf{INC} algorithm can be achieved if one considers the primal
graph $G_\mathit{prim}(\Pi)$. Recall that by definition, a tree decomposition
must contain each edge of the original graph in some node bag. Since the primal
graph contains a clique between all atoms that participate in a rule $r$, it
follows from the connectedness condition of a tree decomposition that there will
be at least one node whose bag contains all the atoms of rule $r$. 
We denote by $\hat{R}_t$ all rules induced in that way by the bag $\chi(t)$ of node $t$.
Thus, rule
satisfaction can be checked immediately, and separate sets to keep track of
satisfied rules are no longer needed in the tuple structure. For a given tree
node $t$, the set $\tau_t$ thus contains simplified tuples of the form $\langle
M, C \rangle$, where $M$ is the same as for \textsf{INC}, and $C$ is a set of
sets $A \subset M$, that are again the same as for \textsf{INC}. Otherwise, the
intuition of the algorithm below is similar to the one of \textsf{INC}. If
primal and incidence graph have similar treewidth, the \textsf{PRIM} algorithm
should benefit from this simplified logic. The \textsf{PRIM} algorithm is given
in Figure~\ref{fig:prim}.
Again, checking whether an answer set exists is equivalent to checking whether
the tuple $\langle \emptyset, \emptyset \rangle$ exists at the root node of the
tree decomposition.

\begin{figure}
\small
\begin{eqnarray*}
  AI(a): & \{ \langle & M,\\
  && \{ A \mid A \in C, \forall r \in \hat{R}_t: A \models_t r^M \}\rangle,\\
  & \mid & \langle M, C \rangle \in \tau_{t'}, \forall r \in \hat{R}_t: M
  \models_t r\} \cup\\
  & \{ \langle & M,\\
  && \{ A \mid A \in C, \forall r \in \hat{R}_t: A \models_t r^M \} \cup\\
  && \{ A \cup \{ a \} \mid A \in C, \forall r \in \hat{R}_t: A \cup \{ a \} \models_t
  r^M \} \cup\\
  && \{ M'' \mid M'' = M', \forall r \in \hat{R}_t: M'' \models_t r^M \}\rangle\\
  & \mid & \langle M', C \rangle \in \tau_{t'}, M = M' \cup \{ a \} \}\\
  \\[-2ex]
  AR(a): & \{ \langle & M \setminus \{ a \}, \{ A \setminus \{ a \} \mid A \in C
  \}\rangle\\
  & \mid & \langle M, C \rangle \in \tau_{t'} \}\\
  \\[-2ex]
  JOIN: & \{ \langle & M, (C' \cap C'') \cup (\{M\} \cap C'') \cup (C' \cap \{M\}) \rangle\\
  & \mid & \langle M, C' \rangle \in \tau_{t'}, \langle M, C'' \rangle \in
  \tau_{t''} \}\\
  \\[-2ex]
  LEAF: && \{ \langle \emptyset, \emptyset \rangle \}
\end{eqnarray*}
\normalsize
\vspace{-0.4cm}
\caption{\textsf{PRIM} Algorithm.}
\label{fig:prim}
\end{figure}

The idea underlying the \textsf{INVPRIM} algorithm is to save in each tuple,
instead of the set $C$ the set $\overline{C}$ of \emph{inverse} certificates,
i.e.\ sets $A \subset M$ that are surely not counter-models w.r.t.\ the reduct
$\Pi^M$. We leave the straightforward adaptation of the above \textsf{PRIM}
algorithm to the interested reader.

\paragraph{Counting.} The three algorithms presented above, as-is, do not allow
for model counting. However, a simple modification allows this. To this end,
associate with each tuple $\overline{t}$ a number $n(\overline{t})$. For tuples
in leaf nodes, set this number to $1$. For tuples in join nodes, let
$n(\overline{t}) = n(\overline{t'}) \cdot n(\overline{t''})$, where $t'$ ($t''$)
is the left (right) child's tuple that gave rise to $\overline{t}$. For
introduction and removal nodes, let $f$ be the surjective function that maps a
child tuple $\overline{t'}$ to tuple $\overline{t}$, according to the algorithms
given above. Then, let $n(\overline{t}) = \mathit{sum}_{\overline{t'} \in
f^{-1}(\overline{t})} n(\overline{t'})$ (i.e.\ if two tuples map to the same
tuple, their counts are summed up).


\paragraph{Correctness and Runtime.} The correctness proof of these algorithms
is rather tedious, as each node type needs to be investigated separately.
However, it is not difficult to see that a tuple at a node $t$ guarantees that
there exists a model for the ASP sub-program induced by the subtree rooted at
$t$, proving soundness. Conversely, it can be shown that each candidate answer
set is indeed evaluated while traversing the tree decomposition, which proves
completeness. Regarding the theoretical runtime bounds, the algorithms all work
in time $O(2^{2^w} \cdot n)$, where $w$ is the width of the underlying tree
decomposition, and $n$ is the size of $\Pi$.

An interesting observation is that, by dropping all the logic concerning the
certificates from the above algorithms, one obtains a pure satisfiability
checking algorithm, similar to those proposed in~\cite{jda:SamerS10}.


\newcommand\abs[1]{\left|#1\right|}%
\newcommand{\SB}{\{\,}%
\newcommand{\SM}{\,{:}\,}%
\newcommand{\SE}{\,\}}%
\def\hy{\hbox{-}\nobreak\hskip0pt} 
\def\hyph{-\penalty0\hskip0pt\relax}
\newcommand{\solver}[1]{\mbox{\text{#1}}\xspace}
\newcommand{\depqbf}{\solver{DepQBF}}
\newcommand{\depqbfz}{\solver{DepQBF0}}
\newcommand{\sharpsat}{\solver{SharpSAT}}
\newcommand{\dynasp}[1]{\ensuremath{\solver{DynASP}(#1)}}

\section{Experimental Evaluation}\label{sec:evaluation}
We performed experiments to evaluate the efficiency of our approach
and its various algorithm configurations (\algo{PRIM}, \algo{INVPRIM},
\algo{INC}) on programs where we can heuristically find a
decomposition of small width reasonably fast. In fact, programs of
small width exist in practice as real-world graphs often admit tree
decompositions of small width. 
Further, we compared our approach with a modern \ASP solver, recent
\numSAT solvers, and a \QBF solver. The solvers tested include our own
prototype implementation, which we refer to as DynASP, and the existing solvers
Cachet~1.21~\cite{sat:SangBBKP04} (a SAT model counter), \depqbfz\footnote{Since
  \depqbf~\cite{LonsingBiere10} does not support counting by default,
  we implemented a naive counting approach into \depqbf using methods
  described in~\cite{Lonsing15}, which we call \depqbfz.} (a QBF solver),
Clasp~3.1.4~\cite{GebserKaufmannSchaub12a} (an ASP solver),
and \sharpsat~12.08~\cite{sat:Thurley06} (a SAT model counter).
%

\begin{figure*}
\centering
\includegraphics[scale=0.39]{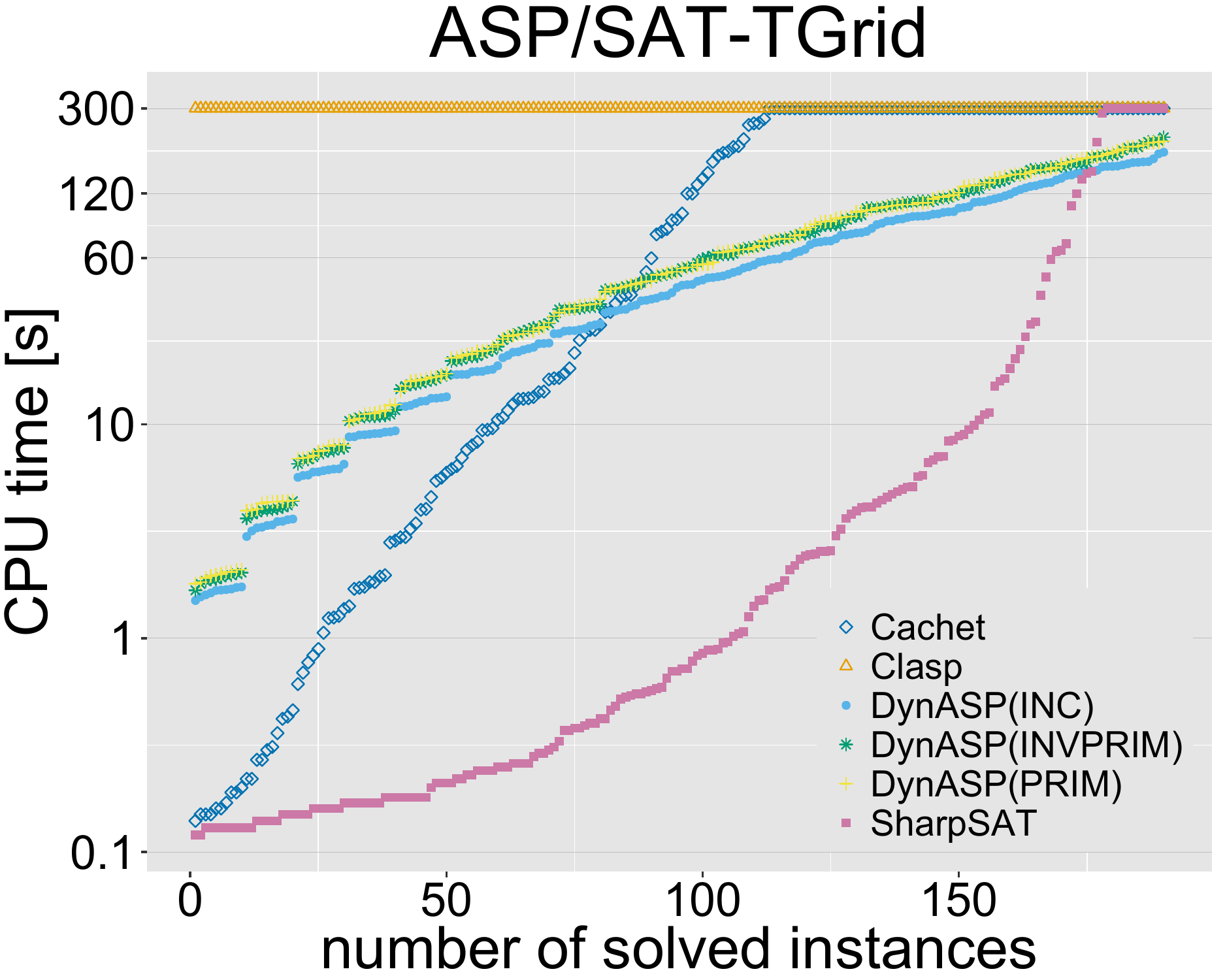}\\[1ex]
\includegraphics[scale=0.39]{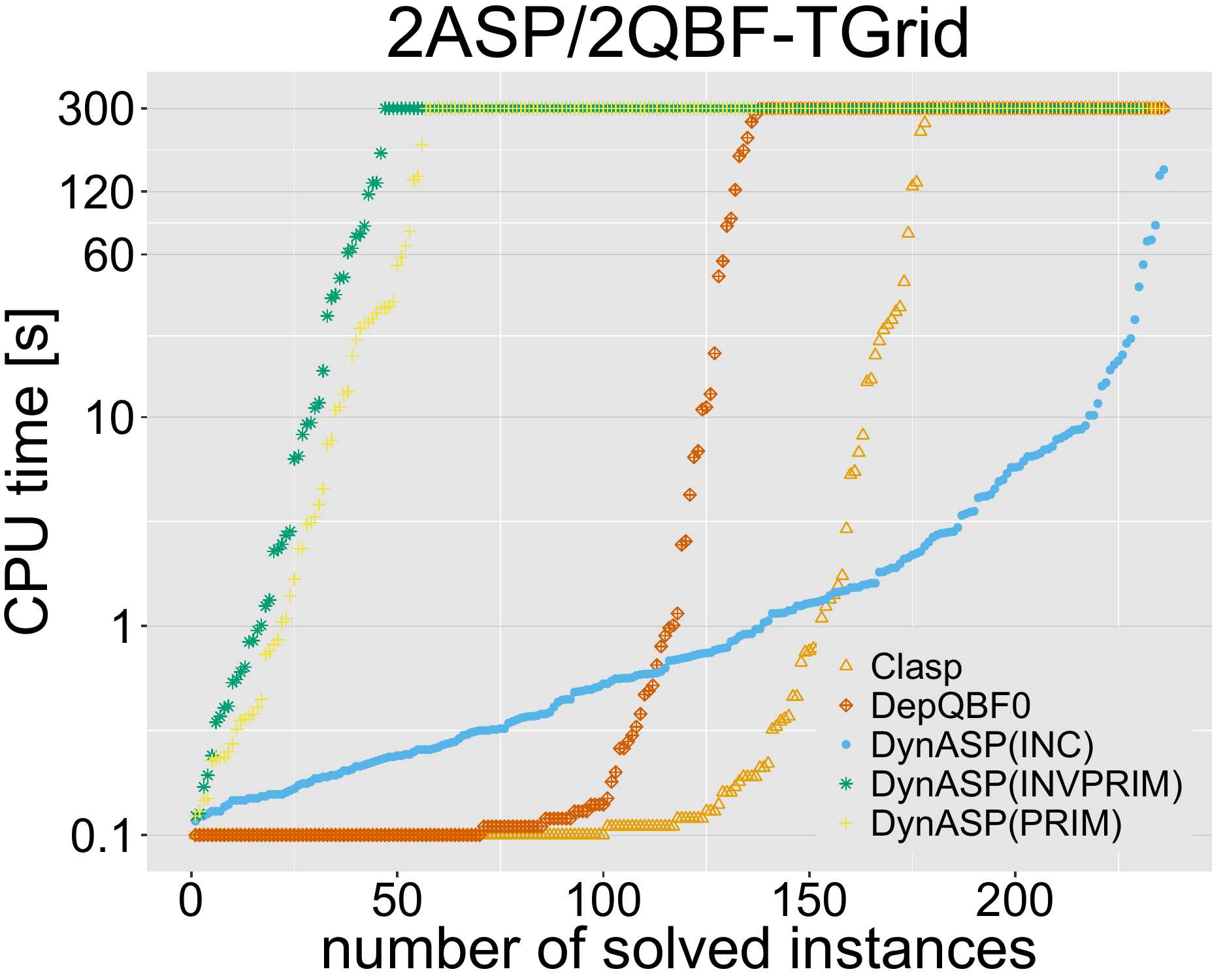}\\[1ex]
\includegraphics[scale=0.39]{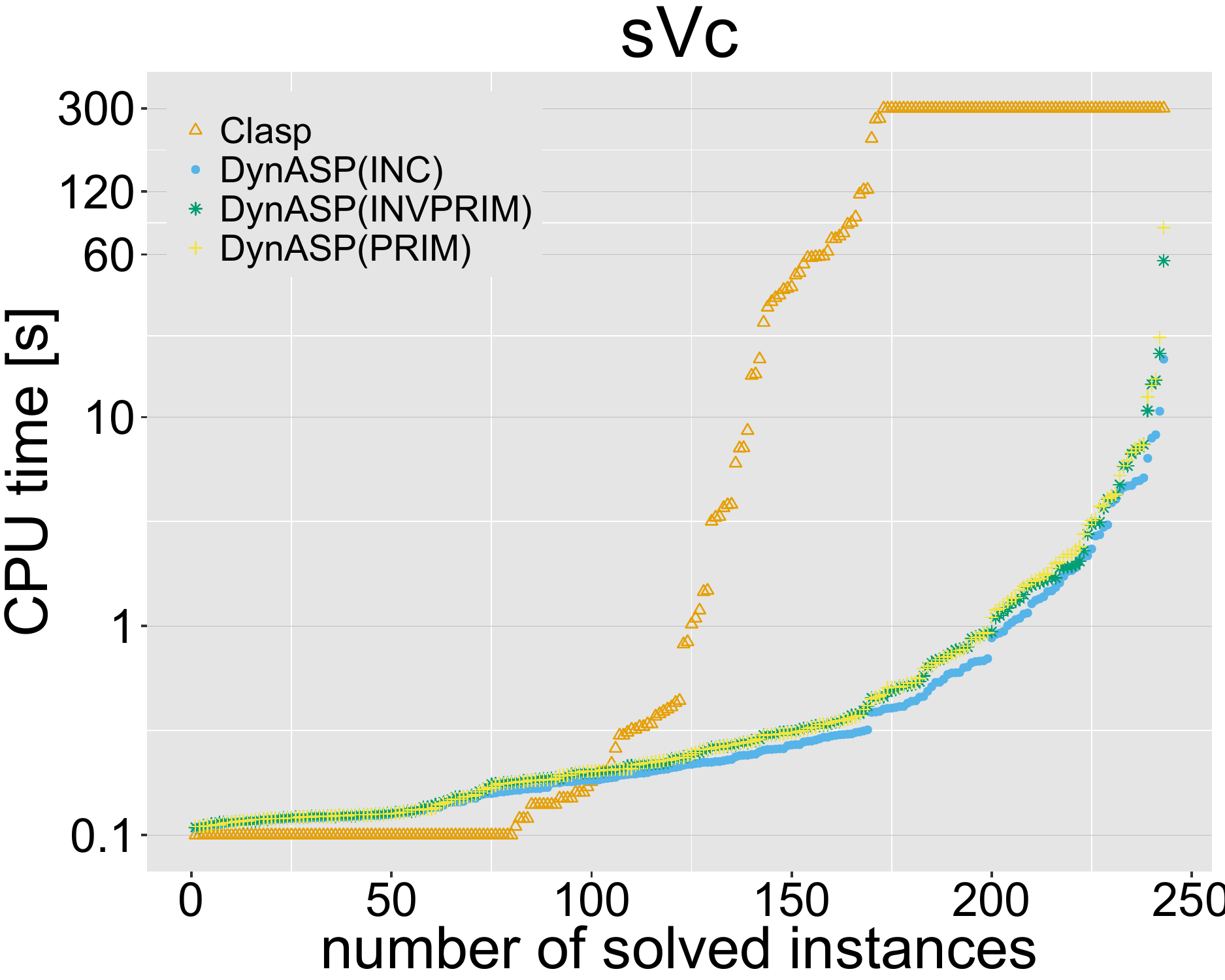}

\caption{Visualization of benchmark results of randomly generated
  instance sets and a selected real-world instance set. 
}
  \label{fig:random}
\end{figure*}

We used both random and structured instances for benchmark
sets, of which we give a description below. The random instances (\pname{Sat-TGrid},
\pname{2QBF-TGrid}, \pname{ASP-TGrid}, \pname{2ASP-TGrid}) were
designed to have a high number of variables and solutions, but with
certain probability a treewidth larger than some fixed~$k$. Therefore,
let $k$ and $\ell$ be some positive integers and $p$ a rational number
such that $0<p\leq 1$.
An instance~$F$ of \pname{Sat-TGrid$(k,l,p)$} consists of the
set~$V=\SB (1,1),\ldots, (1,\ell),(2,\ell),\ldots, (k,\ell) \SE$ of
variables and with probability~$p$ for each variable~$(i,j)$ such that
$1<i\leq k$ and $1<j\leq \ell$ a clause $s_1(i,j)$, $s_2(i - 1,j)$,
$s_3(i,j - 1)$, a clause $s_4(i,j)$, $s_5(i - 1,j)$, $s_6(i-1,j - 1)$,
and a clause $s_7(i,j)$, $s_8(i - 1,j - 1)$, $s_9(i,j - 1)$ where
$s_i \in \{-,+\}$ is selected with probability one half. In that way,
such an instance has an underlying dependency graph that consists of
various triangles forming for probability~$p=1$ a graph that has a
grid as subgraph.  Let $q$ be a rational number such that $0<q\leq
1$. An instance of the set~\pname{2Qbf\hy TGrid$(k,l,p,q)$} is of the
form~$\exists V_1. \forall V_2. F$ where a variable belongs to~$V_1$
with probability~$q$ and to $V_2$ otherwise. Instances of the
sets~\pname{ASP-TGrid} or \pname{2ASP-TGrid} have been constructed in
a similar way, however, as an \ASP program instead of a formula. Note
that the number of answer sets and the number of satisfiable
assignments correspond.
We fixed the parameters to $p=0.85$, $k=3$, and
$l\in \{40,80,\ldots,400\}$ to obtain instances that have with high
probability a small fixed width, a high number of variables and
solutions. Further, we took fixed random seeds and generated 10
instances to ensure a certain randomness.
%
The structured instances model various graph problems (\pname{2Col},
\pname{3Col}, 
\pname{Ds}, \pname{St} \pname{cVc}, \pname{sVc}) on real world mass
transit graphs of 82 cities, metropolitan areas, or countries (e.g.,
Beijing, Berlin, Shanghai, and Singapore).
The graphs have been extracted from publicly available mass transit
data feeds~\cite{gtfs} using gtfs2graphs~\cite{Fichte16c} and split
by transportation type,~e.g., train, metro, tram, combinations. We
heuristically computed tree decompositions \cite{Dermaku} and obtained
relatively fast decompositions of small width unless detailed bus
networks were present.
%
%
The encoding for \pname{2Col} counts all minimal sets~$S$ of
vertices such that there are two sets~$F$ and $S$ where no two
neighboring vertices~$v$ and $w$ belong to~$F$ for a given input
graph. The encoding for \pname{3Col} models to count all 3-colorings.
The encoding for \pname{Ds} models to count all minimal dominating
sets. The encoding for \pname{St} models to count all Steiner trees.
The encoding for \pname{cVc} asks to count all minimal vertex
covers. The encoding for \pname{sVc} models to count all
subset-minimal vertex covers.
%
%

We ran the experiments on an Ubuntu 12.04 Linux cluster of 3 nodes
with two AMD Opteron 6176 SE CPUs of 12 physical cores each at 2.3Ghz
clock speed and 128GB RAM. Input instances were given to the solvers
via a shared memory file system.
During a run we limited the available memory to 4GB RAM and the CPU
time to 
300 seconds.
We used default options for cachet and \sharpsat,
option \mbox{``--qdc''} for \depqbfz, option ``--stats=2
--opt-mode=optN -n 0 --opt-strategy=bb/usc -q'' for clasp, 
and will refer to the different variants of our prototype implementation as
\dynasp{\algo{PRIM}}, \dynasp{\algo{INVPRIM}} and \dynasp{\algo{INC}}.
Since we cannot expect to solve instances of high treewidth
efficiently, we restricted the instances such that we were able to
heuristically find a decomposition of width smaller than 20 within 60
seconds. 

In order to draw conclusions about the efficiency of our approach, we
mainly inspected the (total cpu) running time\footnote{The runtime for
  \dynasp{\cdot} includes decomposition times. Note that we randomly
  generated three fixed seeds for the decomposition computation
  to allow a certain variance in decomposition
  features~\cite{AbseherEtAl15a}. When evaluating the results, we
  constructed the average on the seeds per instance.
} and number of timeouts on the random and structured benchmark
sets. %
Figure~\ref{fig:random} illustrates the running times of the solvers
on the various random instance sets and a selected structured
instance set as a cactus plot. Table~\ref{tab:real_world} reports on
the average running times, number of solved instances, and number of
timeouts of the solvers on the considered structured instance sets.
%

%
%
%
\pname{SAT-TGrid} and \pname{Asp-TGrid}:
Cachet solved 125 instances.
%
Clasp always timed out for both configurations (branch and bound; and
unsatisfiable core). A reason could be the high number of solutions as
Clasp counts the models by enumerating them.
\dynasp{\cdot} solved each instance within at most 270 seconds (on
average 67 seconds). The best configuration with respect to runtime
was \algo{PRIM}. However, the running times of the different
configurations were close. We observed as expected a sub-polynomial
growth in the runtime with an increasing number of solutions.
\sharpsat timed out on 3 instances and ran into a memory out on 7
instances, but solved most of the instances quite fast. Half of the
instances were solved within 1 seconds and more than 80\% of the
instances within 10 seconds, and about 9\% of the instances took more
than 100 seconds. The number of solutions does not have an impact on
the runtime of \sharpsat.
\sharpsat was the fastest solver in total, however, \dynasp{\cdot}
solved all instances. The results are illustrated in the two left
graphs of Figure~\ref{fig:random}.

\pname{2QBF-TGrid} and \pname{2ASP-TGrid}: Clasp solved more than half
of the instances in less than 1 second, however, timed out on 59
instances. \depqbfz shows a similar behavior as Clasp, which is not
surprising as both solvers count the number of solutions by
enumerating them and hence the number of solutions has a significant
impact on the runtime of the solver. However, Clasp is throughout
faster than \depqbfz.
%
DynASP(\algo{INC}) solved half of the instances within less than 1
second, about 92\% of the instances within less than 10 seconds, and
provided solutions also if the instance had a large number of answer
sets. DynASP(\algo{INVPRIM}) and DynASP(\algo{PRIM}) quickly produced
timeouts due to a significantly larger width of the computed
decompositions.

Structured instances: 
Clasp solved most of the structured instances reasonably
fast. However, the number of solutions has again, similar to the
random setting, a significant impact on its performance. If the
instance has a small number of solutions, then Clasp yields the number
almost instantly. If the number of solutions was very high, then Clasp
timed out.
\dynasp{\cdot} solved for each set but the set~\pname{St} more than
80\% of the instances in less than 1 second and the remaining
instances in less than 100 seconds. For \pname{St} the situation was
different. Half of the instances were solved in less than 10 seconds
and a little less than the other half timed out.  Similar to the
random setting, \dynasp{\cdot} ran still fast on instances with a
large number of solutions. Whenever the instance had relatively few
solutions Clasp was faster, otherwise \dynasp{\cdot} (e.g.,
\pname{sVc}) was faster.


\newcommand{\blahtab}[1]{{\tiny{#1}}}
\begin{table}[t]
  \centering
  \small
  \begin{tabular}{@{}l@{\hspace{0.1em}}|@{\hspace{0.2em}}
	r@{\hspace{0.2em}}r@{\hspace{0.3em}}
	r@{}			  r@{\hspace{0.3em}}
	r@{\hspace{0.2em}}r@{\hspace{0.3em}}
	r@{\hspace{0.2em}}r@{\hspace{0.3em}}
	r@{\hspace{0.2em}}r@{\hspace{0.3em}}
	r@{\hspace{0.2em}}r@{}}
    \toprule 
            & \multicolumn{2}{c}{\pname{2Col}} & 
					\multicolumn{2}{c}{\pname{3Col}} & 
					\multicolumn{2}{c}{\pname{Ds}} & 
					\multicolumn{2}{c}{\pname{St}} & 
					\multicolumn{2}{c}{\pname{cVc}} & 
					\multicolumn{2}{c}{\pname{sVc}}\\ 
    \midrule
    {\blahtab{Clasp}} &  31.72 & (21) & 0.10 &(0)  & 8.99 & (3) & 0.21 & (0)  & 29.88 & (21) & 98.34 & (71) \\
    \blahtab{INC} & 1.43 &(0) & 0.58 &(0) & 0.54 & (0) & 115.02 & (498) & 0.68 & (0) & 0.78 & (0) \\
    \blahtab{INVP} & 1.50 &(0) & 0.47 &(0) & 0.79 & (0) & 91.92 & (248) & 0.99 & (0) & 1.15 & (0) \\
    \blahtab{PRIM} & 1.54 &(0) & 0.53 &(0) & 0.68 & (0) & 79.36 & (221) & 0.99 & (0) & 1.30 & (0)\\
    \bottomrule
  \end{tabular}\vspace{1em}
  \caption{Runtime results on real-world graph instance sets and number of timeouts in brackets. Runtimes given in sec.}
  \label{tab:real_world}
\end{table}    

The empirical results of the benchmarks confirm that our DynASP prototype
works reasonably fast under the assumption that the input
instance has small treewidth. The comparison to state-of-the-art \ASP
and \QBF solvers shows that our solver has an advantage if we have to
count many solutions, whereas Clasp and \depqbfz perform well if the
number of instances is relatively small. However, \dynasp{\cdot} is still
reasonably fast on structured instances with few solutions as it
yields the number of solutions mostly within less than 10 seconds.  We
observed that \dynasp{\algo{INC}} seems to be the overall best
solving algorithm in our setting, which indicates that the smaller treewidth
obtained by decomposing the incidence graph generally outweighs the benefits of
simpler solving algorithms for the primal graph.
A comparison to recent \#SAT solvers
suggests that dedicated \#SAT algorithms are somewhat faster on random
\SAT formulas of small treewidth than our decomposition based
approach, which is, however, not particularly surprising since our
implementation is equipped to handle the full ASP
semantics. 
The results indicate that our approach seem to be suitable for
practical use, at least for certain classes of instances with low
treewidth, and hence could fit into a portfolio-based solver.



\section{Conclusions}\label{sec:conclusions}

In this paper, we have presented several dynamic programming algorithms for
counting answer sets of logic programs, and compared a prototype implementation
to existing solvers. For large instances of low treewidth, our implementation
proved to be competitive both against classical ASP solvers that need to
materialize all answer sets in order to count them, as well as specific counting
algorithms developed for SAT. These promising results confirm the importance of
evaluating parameterized algorithms in practice~\cite{paracompnews:Gutin15}.
Future work includes extending our algorithms to weighted model counting, to
solve~e.g., the Bayesian inference problem.

\subsection*{Acknowledgements}
The authors gratefully acknowledge support by the Austrian Science
Fund (FWF), Grant Y698. The first author is also affiliated with the
Institute of Computer Science and Computational Science at University
of Potsdam, Germany.




\end{document}